%% file: kss.tex
\documentclass[a4paper,UKenglish,cleveref, autoref, thm-restate]{lipics-v2019}

\input{sections/_packages_and_commands}
\input{sections/_algo}

\title{Faster Algorithms for $k$\textsc{-Subset Sum} and Variations} 

\titlerunning{Faster Algorithms for $k$\textsc{-Subset Sum} and Variations} 

\author{Antonis Antonopoulos}{National Technical University of Athens, Greece }{aanton@corelab.ntua.gr}{https://orcid.org/0000-0002-1368-6334}{}

\author{Aris Pagourtzis}{National Technical University of Athens, Greece }{pagour@cs.ntua.gr}{https://orcid.org/0000-0002-6220-3722}{Research supported in part by the PEVE 2020 basic research support programme of the National Technical University of Athens.}

\author{Stavros Petsalakis}{National Technical University of Athens, Greece }{spetsalakis@corelab.ntua.gr}{https://orcid.org/0000-0001-7825-2839}{Research supported in part by the PEVE 2020 basic research support programme of the National Technical University of Athens.}

\author{Manolis Vasilakis}{National Technical University of Athens, Greece }{mvasilakis@corelab.ntua.gr}{https://orcid.org/0000-0001-6505-2977}{}

\authorrunning{A. Antonopoulos, A. Pagourtzis, S. Petsalakis and M. Vasilakis} 

\Copyright{Antonis Antonopoulos, Aris Pagourtzis, Stavros Petsalakis and Manolis Vasilakis} 



\category{} 

\relatedversion{} 

\supplement{}

\acknowledgements{}

\nolinenumbers 

\hideLIPIcs  

\EventEditors{John Q. Open and Joan R. Access}
\EventNoEds{2}
\EventLongTitle{42nd Conference on Very Important Topics (CVIT 2016)}
\EventShortTitle{CVIT 2016}
\EventAcronym{CVIT}
\EventYear{2016}
\EventDate{December 24--27, 2016}
\EventLocation{Little Whinging, United Kingdom}
\EventLogo{}
\SeriesVolume{42}
\ArticleNo{23}

\begin{document}

\maketitle

\begin{abstract}
    We present new, faster pseudopolynomial time algorithms for the $k$\textsc{-Subset Sum} problem, defined as follows: given a set $Z$ of $n$ positive integers and $k$ targets $t_1, \ldots, t_k$, determine whether there exist $k$ disjoint subsets $Z_1,\dots,Z_k \subseteq Z$, such that $\Sigma(Z_i) = t_i$, for $i = 1, \ldots, k$.
    Assuming $t = \max \{ t_1, \ldots, t_k \}$ is the maximum among the given targets, a standard dynamic programming approach based on Bellman's algorithm~\cite{Bellman} can solve the problem in $O(n t^k)$ time.
    We build upon recent advances on \textsc{Subset Sum} due to Koiliaris and Xu~\cite{KoiliarisX19} and Bringmann~\cite{Bringmann17} in order to provide faster algorithms for $k$\textsc{-Subset Sum}.
    We devise two algorithms: a deterministic one  of time complexity $\tilde{O}(n^{k / (k+1)} t^k)$ and a randomised one of $\tilde{O}(n + t^k)$ complexity.
    Additionally, we show how these algorithms can be modified in order to incorporate cardinality constraints enforced on the solution subsets.
    We further demonstrate how these algorithms can be used in order to cope with variations of $k$\textsc{-Subset Sum}, namely \textsc{Subset Sum Ratio}, $k$\textsc{-Subset Sum Ratio} and \textsc{Multiple Subset Sum}.
\end{abstract}

\input{sections/1_introduction}
\input{sections/2_preliminaries}
\input{sections/3_0_kss}
\input{sections/3_1_kss_koiliaris}
\input{sections/3_2_kss_bringmann}
\input{sections/4_0_fixed_cardinality}
\input{sections/4_1_fixed_koiliaris}
\input{sections/4_2_fixed_bringmann}
\input{sections/5_further_applications}
\input{sections/6_future_work}



\bibliography{kSS}
\end{document}

%% file: sections/_packages_and_commands.tex
\newcommand{\MSS}{\textsc{Multiple Subset Sum}}
\newcommand{\kSS}{$k$\textsc{-Subset Sum}}
\newcommand{\SSR}{\textsc{Subset Sum Ratio}}
\newcommand{\kSSR}{$k$\textsc{-Subset Sum Ratio}}
\newcommand{\MK}{\textsc{Multiple Knapsack}}
\newcommand{\SubS}{\textsc{Subset Sum}}
\newcommand{\ESS}{\textsc{Equal Subset Sum}}
\newcommand{\ColorCoding}{{\normalfont\texttt{ColorCoding}}}
\newcommand{\ColorCodingLayer}{{\normalfont\texttt{ColorCodingLayer}}}
\newcommand{\Mod}[1]{\ (\mathrm{mod}\ #1)}

\usepackage{hyperref}

\usepackage{amsfonts}
\usepackage{amsmath}

\usepackage{enumitem}

%% file: sections/_algo.tex
\usepackage{algorithm}
\usepackage{algpseudocode}

\algtext*{EndIf}
\algtext*{EndFor}

\algrenewcommand\algorithmicrequire{\textbf{Input :}}
\algrenewcommand\algorithmicensure{\textbf{Output :}}

\algnewcommand{\lIf}[2]{
  \State \algorithmicif\ #1\ \algorithmicthen\ #2}

\algnewcommand{\lElsIf}[2]{
  \State \algorithmicelse\ \algorithmicif\ #1\ \algorithmicthen\ #2}

\algnewcommand{\lElse}[1]{
  \State \algorithmicelse\ #1}
  
\algnewcommand{\Ret}[1]{
  \State \algorithmicreturn\ #1}

\algnewcommand{\IfThenElse}[3]{
  \State \algorithmicif\ #1\ \algorithmicthen\ #2\ \algorithmicelse\ #3}

%% file: sections/1_introduction.tex
\section{Introduction}
\label{sec:introduction}

One of Karp's~\cite{Karp72} $21$ NP-complete problems, \SubS\ has seen astounding 
progress with respect to its pseudopolynomial time solvability over the last few years.
Koiliaris and Xu~\cite{KoiliarisX19} and Bringmann~\cite{Bringmann17} have presented pseudopolynomial algorithms resulting in substantial improvements over the long-standing standard approach of Bellman~\cite{Bellman}, and the improvement by Pisinger~\cite{Pisinger99}.

\ESS\ is a closely related algorithmic problem with applications in computational biology~\cite{CieliebakEP03,CieliebakE04}, computational social choice~\cite{LiptonMMS04}, and cryptography~\cite{vol}, to name a few.
Additionally, it is related to important theoretical concepts such as the complexity of search problems in the class \textsf{TFNP}~\cite{Papadimitriou94}.
This problem can be solved in $\tilde{O}(n+t)$ time\footnote{$\tilde{O}$ notation ignores polylogarithmic factors.} (where $t$ is a bound for the sums we are interested in) by a simple extension of Bellman's~\cite{Bellman} algorithm, as it only asks for two disjoint subsets of equal sum $s$, for some $s \leq t$, and as such can exploit the following \emph{disjointness property}:
if there exists a pair of sets with equal sum $s \leq t$, then by removing their common elements, one can produce two disjoint sets of equal sum $s' < s \leq t$.
Therefore, the algorithm terminates as soon as a sum is produced for a second time and by using this property as well as an appropriate tree data structure, we can bound the total running time by the number of all the possible distinct subset sums up to $t$, hence obtaining an $\tilde{O} (n + t)$ algorithm.

A generalisation of \SubS\ is the problem that asks for $k$ disjoint subsets the sums of which are respectively equal to targets $t_i,i=1\dots k$, henceforth referred to as the \kSS\ problem. 
One can see that even in the case of $k=2$ and $t_1=t_2$, the problem seems to escalate in complexity, since the aforementioned disjointness property does not hold.
Particular interest can be found in the special case when the sum of targets equals the sum of the input elements, that is, we ask for a partition of the input set to subsets of particular sums.
Another special case, which finds applications in fair allocation problems, is when we want to partition the input set to a number of subsets of equal sum, i.e.\ all $t_i$ are equal to the sum of the input elements divided by $k$.

\subsection{Related work}

\ESS\ as well as its optimisation version called \SSR~\cite{BazganST02} are closely related to problems appearing in many scientific areas.
Some examples include the Partial Digest problem, which comes from computational biology~\cite{CieliebakEP03,CieliebakE04}, the
allocation of individual goods~\cite{LiptonMMS04}, tournament construction~\cite{khanphd}, and a variation of \SubS, called Multiple Integrated Sets SSP, which finds applications in the field of cryptography~\cite{vol}.
Furthermore, it is related to important concepts in theoretical computer science; for example, a restricted version of \ESS\ lies in a subclass of  the complexity class $\mathsf{TFNP}$, namely in $\mathsf{PPP}$~\cite{Papadimitriou94}, a class consisting of search problems that always have a solution due to some pigeonhole argument, and no polynomial time algorithm is known for this restricted version.

\ESS\ has been proven NP-hard by Woeginger and Yu~\cite{WoegingerY92} and several variations have been proven NP-hard by Cieliebak \textsl{et al.} in~\cite{CieliebakEP03b,CieliebakEPS08}.
A 1.324-approximation algorithm has been proposed for \textsc{Subset Sum Ratio} in~\cite{WoegingerY92} 
and several FPTASs appeared in~\cite{BazganST02,Nanongkai13,MelissinosP18}, the fastest so far being the one in~\cite{MelissinosP18} of complexity $O(n^4/\varepsilon)$. 

As far as exact algorithms are concerned, recent progress has shown that \ESS\ can be solved probabilistically in $O^{*}(1.7088^n)$ time~\cite{MuchaNPW19}, which is
faster than a standard ``meet-in-the-middle'' approach yielding an $O^{*}(3^{n/2}) \le O^{*}(1.7321^n)$ time algorithm.
Additionally, as explained earlier, an exact solution can be obtained in pseudopolynomial $\tilde{O}(n+t)$ time using an extension of Bellman's~\cite{Bellman} algorithm, when an upper bound $t$ on the sums is provided in the input.
However, these techniques do not seem to apply to \kSS, mainly due to the fact that we cannot assume the minimality of the involved subsets.
To the best of our knowledge, no pseudopolynomial time algorithm substantially faster than the standard $O(n t^k)$ dynamic programming approach was known for \kSS\ prior to this work. 

These problems are tightly connected to \SubS, which has seen impressive advances recently, due to Koiliaris and Xu~\cite{KoiliarisX19} who gave a deterministic $\tilde{O}(\sqrt{n}t)$ algorithm, where $n$ is the number of input elements and $t$ is the target, and by Bringmann~\cite{Bringmann17} who gave a $\tilde{O}(n + t)$ randomised algorithm.
Jin and Wu proposed a simpler randomised algorithm~\cite{JinWu} achieving the same bounds as~\cite{Bringmann17}, which however seems to only solve the decision version of the problem.
Recently, Bringmann and Nakos~\cite{BringmannN20} have presented an 
$O(\rvert\mathcal{S}_t(Z) \rvert^{4/3}poly(\log t))$ algorithm, where $\mathcal{S}_t(Z)$ is the set of all subset sums of the input set $Z$ that are smaller than $t$, based on top-$k$ convolution.

\MSS\ is a special case of \MK, both of which have attracted considerable attention.
Regarding \MSS, Caprara \textsl{et al.} present a PTAS for the case where all targets are the same~\cite{CapraraKP00}, and subsequently in~\cite{CapraraKP03} they introduce a $3/4$ approximation algorithm. 
The \MK\ problem has been more intensively studied in recent years as applications for it arise naturally (in fields such as transportation, industry, and finance, to name a few).
Some notable studies on variations of the problem are given by Lahyani \textsl{et al.}~\cite{LahyaniCKC19} and Dell'Amico \textsl{et al.}~\cite{DellAmicoDIM19}. Special cases and variants of \MSS, such as the \kSS\ problem, have been studied in~\cite{CieliebakEPS08,CieliebakEP03b} where simple pseudopolynomial algorithms were proposed.

\subsection{Our contribution}

We first present two algorithms for \kSS: a deterministic one of complexity $\tilde{O}(n^{k/(k+1)} t^k)$ and a randomised one of complexity $\tilde{O}(n + t^k)$.
Then, we describe the necessary modifications required in order to solve \kSS\ in the case where there exist some cardinality constraints regarding the solution subsets.
We subsequently show how these ideas can be extended to solve the decision versions of \SSR, \kSSR\ and \MSS.

Our algorithms extend and build upon the algorithms and techniques proposed by Koiliaris and Xu~\cite{KoiliarisX19} and Bringmann~\cite{Bringmann17} for \SubS.
In particular, we make use of FFT computations, modular arithmetic and color-coding, among others. 

We start by presenting some necessary background in Section~\ref{sec:prelim}.
Then, we present the two \kSS\ algorithms in Section~\ref{sec:kss}, followed by the two algorithms for the cardinality constrained version of the problem in Section~\ref{sec:kss_cardinality}.
We next show how these algorithms can be used to efficiently decide multiple related subset problems.
Finally, we conclude the paper by presenting some directions for future work.

%% file: sections/2_preliminaries.tex
\section{Preliminaries}
\label{sec:prelim}

\subsection{Notation}

We largely follow the notation used in~\cite{Bringmann17} and~\cite{KoiliarisX19}.

\begin{itemize}
    
    \item Let $[x] = \{ 0, \ldots, x \}$ denote the set of integers in the interval $[0, x]$.
    
    \item Given a set $Z \subseteq \mathbb{N}$, we denote:
    \begin{itemize}
        \item the sum of its elements by $\Sigma (Z) = \sum_{z \in Z} z$.

        \item the \emph{characteristic polynomial of Z} by $f_Z (x)= \sum_{z \in Z} x^z$.
        
        \item the \emph{k-modified characteristic polynomial of Z} by $f_Z^k (\vec{x}) = \sum_{z \in Z} \sum_{i=1}^k x_i^z$, where $\vec{x} = (x_1, \ldots, x_k)$.

        \item the \emph{set of all subset sums of Z up to t} by $\mathcal{S}_t (Z) = \{ \Sigma (X) \mid X \subseteq Z \} \cap [t]$.

    \end{itemize}
    
    \item For two sets $X,Y \subseteq \mathbb{N}$, let 
    \begin{itemize}
        \item $X \oplus Y = \{x + y \mid x \in X \cup \{0\}, y \in Y \cup \{0\} \}$ denote the \emph{sumset} or \emph{pairwise sum} of sets $X$ and $Y$.
        
        \item $X \oplus_t Y = (X \oplus Y) \cap [t]$ denote the \emph{t-capped sumset} or \emph{t-capped pairwise sum} of sets $X$ and $Y$. Note that $t > 0$.
    \end{itemize}
    
    \item The pairwise sum operations can be extended to sets of multiple dimensions.
    Formally, let $X,Y \subseteq \mathbb{N}^k$.
    Then, $X \oplus Y = \{ (x_1 + y_1, \ldots, x_k + y_k) \}$, where $(x_1, \ldots, x_k) \in X \cup \{0\}^k$ and $(y_1, \ldots, y_k) \in Y \cup \{0\}^k $.

\end{itemize}

\subsection{Using FFT for \SubS}
Given two sets $A, B \subseteq \mathbb{N}$ and an upper bound $t>0$, one can compute the t-capped pairwise sum set $A \oplus_t B$ using FFT to get the monomials of the product $f_A \cdot f_B$ that have exponent $\leq t$.
This operation can be completed in time $O(t \log t)$.
Observe that the coefficients of each monomial $x^i$ represent the number of pairs $(a, b)$ that sum up to $i$, where $a \in A \cup \{ 0 \}$ and $b \in B \cup \{ 0 \}$.

Also note that an FFT operation can be extended to polynomials of multiple variables.
Thus, assuming an upper bound $t$ for the exponents involved, it is possible to compute $f_A^k \cdot f_B^k$ in $O(t^k \log t)$ time.

\begin{lemma}
Given two sets of points $S , T \subseteq [t]^k$ one can use multidimensional FFT to compute the set of pairwise sums (that are smaller than $t$) in time $O(t^k \log t)$.
\end{lemma}

As shown in~\cite[Ch.~30]{CLRS}, given two set of points $S , T \subseteq [t]^k$ one can pipeline $k$ one-dimensional FFT operations in order to compute a multi-dimensional FFT in time $O(t^k \log t)$.

%% file: sections/3_0_kss.tex
\section{\kSS}
\label{sec:kss}

In this section we propose algorithms that build on the techniques of Koiliaris and Xu~\cite{KoiliarisX19} and Bringmann~\cite{Bringmann17} in order to solve \kSS:  given a set $Z$ of $n$ positive integers and $k$ targets $t_1, \ldots, t_k$, determine whether there exist $k$ disjoint subsets $Z_1, \dots, Z_k \subseteq Z$, such that $\Sigma(Z_i) = t_i$, for $i = 1, \ldots, k$.
Note that $k$ is fixed and not part of the input.
For the rest of this section, assume that $Z = \{ z_1, \ldots, z_n \}$ is the input set, $t_1, \ldots, t_k$ are the given targets and $t = \max \{ t_1, \ldots, t_k \}$.
 
The main challenge in our approach is the fact that the existence of subsets summing up to the target numbers (or any other pair of numbers) does not imply the disjointness of said subsets.
Hence, it is essential to guarantee the disjointness of the resulting subsets through the algorithm.

Note that one can extend Bellman's classic dynamic programming algorithm for \SubS~\cite{Bellman} to solve this problem in $O(n t^k)$ time, as seen in algorithm~\ref{alg:bellman}, where we solve the special case $k = 2$.
Similarly, the algorithm can be extended to apply for any arbitrary $k$.

\begin{algorithm}[htb]
    \caption{$\texttt{Bellman}{(Z,t_1,t_2)}$}
    \label{alg:bellman}
    \begin{algorithmic}[1]
        \Require A set $Z = \{ z_1, \ldots, z_n \}$ of positive integers and targets $t_1, t_2 \leq \Sigma (Z)$.
    	\Ensure True if there are two disjoint subsets $Z_1, Z_2 \subseteq Z$, with $\Sigma (Z_1) = t_1$ and $\Sigma (Z_2) = t_2$, else false.

        \State $t \gets \max \{ t_1, t_2 \}$
        \State initialise table $T[t][t][n] \gets false$ everywhere
        \State $T[0][0][0] \gets true$ \Comment{Initially, only $(\emptyset, \emptyset)$ is possible.}
        \For{$n_i = 1, \ldots, n$}
            \For{$(i, j) \in [t] \times [t]$}
                \If{$T[i][j][n_i - 1] = true$}
                    \State $T[i][j][n_i] = true$ \Comment{$z_{n_i} \notin Z_1 \cup Z_2$}
                    \State $T[i + z_i][j][n_i] = true$ \Comment{$z_{n_i} \in Z_1$}
                    \State $T[i][j + z_i][n_i] = true$ \Comment{$z_{n_i} \in Z_2$}
                \EndIf
            \EndFor
        \EndFor
        \Ret{$T[t_1][t_2][n]$}
    \end{algorithmic}
\end{algorithm}

%% file: sections/3_1_kss_koiliaris.tex
\subsection{Solving \kSS\ in deterministic $\tilde{O}(n^{k/(k+1)} t^k)$ time}
\label{subsec:koiliaris}

In this subsection we show how to decide \kSS\ in $\tilde{O}(n^{k/(k+1)} t^k)$ time, where $t = \max \{t_1,\dots,t_k \}$.
To this end, we extend the algorithm proposed by Koiliaris and Xu~\cite{KoiliarisX19}.
We first describe briefly the original algorithm for completeness and clarity.

The algorithm recursively partitions the input set $S$ into two subsets of equal size and returns all pairwise sums of those sets, along with cardinality information for each sum. This is achieved using FFT for pairwise sums as discussed in section \ref{sec:prelim}.

Using properties of congruence classes, it is possible to further capitalise on the efficiency of FFT for subset sums as follows. 
Given bound $t$, partition the elements of the initial set into congruence classes mod $b$, for some integer $b > 0$.
Subsequently, divide the elements by $b$ and keep their quotient, hence reducing the size of the maximum exponent value from $t$ to $t/b$.
One can compute the $t$-capped sum of the initial elements by computing the $t/b$-capped sum of each congruence class and combining the results of each such operation in a final $t$-capped FFT operation, taking into account the remainder of each congruence class w.r.t. the number of elements (cardinality) of each subset involved.
In order to achieve this, it is necessary to keep cardinality information for each sum (or monomial) involved, which can be done by adding another variable with minimal expense in complexity.

Thus, the overall complexity of FFT operations for each congruence class $l \in \{0, 1, \ldots, b-1\}$ is $O ( (t/b) n_l \log n_l \log t )$, where $n_l$ denotes the number of the elements belonging to said congruence class. 
Combining these classes takes $O(b t \log t)$ time so the final complexity is  $O(t \log t (\frac{n \log n}{b} + b))$.
Setting $b=\lfloor \sqrt{n \log n} \rfloor$ gives $O(\sqrt{n \log n}\,t \log t) = \tilde{O}(\sqrt{n}\,t)$.
After combining the sums in the final step, we end up with a polynomial that contains each realisable sum represented by a monomial, the coefficient of which represents the number of different (not necessarily disjoint) subsets that sum up to the corresponding sum.

Our algorithm begins by using the modified characteristic polynomials we proposed in the preliminaries, thereby representing each $z \in Z$ as $\Sigma_{i=1}^k x_i^z$ in the base polynomial at the leaves of the recursion.
We also use additional dimensions for the cardinalities of the subsets involved; each cardinality is represented by the exponent of some $x_i$, with index $i$ greater than $k$.
We then proceed with using multi-variate FFT in each step of the recursion in an analogous manner as in the original algorithm, thereby producing polynomials with terms $x_1^{t'_1} \ldots x_{k}^{t'_k} x_{k+1}^{c_1} \ldots x_{2k}^{c_k}$, each of which corresponds to a $2k$-tuple of disjoint subsets of sums ${t'_1}, \ldots, {t'_k}$ and cardinalities ${c_1}, \ldots, {c_k}$ respectively.

This results in FFT operations on tuples of $2k$ dimensions, having $k$ dimensions of max value $t/b$, and another $k$ dimensions of max value $n_l$ which represent the cardinalities of the involved subsets, requiring $O(n_l^k(t/b)^k \log (n_l) \log (n_l t/b))$ time for congruence class $l \in \{0, 1, \ldots, b-1\}$.

The above procedure is implemented in Algorithm~\ref{alg:koil1}, while the main algorithm is Algorithm~\ref{alg:koil2}, which combines the results from each congruence class, taking additional $O(b t^k\log t)$ time.

\begin{lemma}
\label{lem:disjointness}
There is a term $x_1^{t'_1} \dots x_k^{t'_k}$ in the polynomial returned by Algorithm~\ref{alg:koil2} if and only if there exist $k$ disjoint subsets $Z_1, \ldots, Z_k \subseteq Z$ such that $\Sigma(Z_i) = t'_i$.
\end{lemma}

\begin{proof}
We observe that each of the terms of the form $x_1^{t'_1}\dots x_k^{t'_k}$ has been produced at some point of the recursion via an FFT operation combining two terms that belong to different subtrees, ergo containing different elements in each subset involved.
As such, $t'_1,\dots ,t'_k$ are sums of disjoint subsets of $Z$.
\end{proof}

\begin{theorem}
Given a set $Z = \{ z_1, \ldots, z_n \}\subseteq \mathbb{N}$ of size $n$ and targets $t_1,\dots ,t_k$, Algorithm~\ref{alg:koil2} can decide \kSS\ in time $\tilde{O}(n^{k/(k+1)} t^k)$, where $t = \max \{t_1,\dots,t_k \}$.
\end{theorem}

\begin{proof}
The correctness of the algorithm stems from Lemma~\ref{lem:disjointness}.

\medskip

\noindent
\emph{Complexity.}
The overall complexity of the algorithm, stemming from the computation of subset sums inside the congruence classes and the combination of those sums, is
\[
    O \left( t^k \log t \left( \frac{n^k \log n}{b^k} + b \right) \right) =
    \tilde{O} (n^{k / (k+1)} t^k),
\]
where the right-hand side is obtained by setting $b = \sqrt[k+1]{n^k \log n}$.

\end{proof}

\begin{algorithm}[htb]
    \caption{$\texttt{DisjointSC}(S,t)$}
    \label{alg:koil1}
    \begin{algorithmic}[1]
        \Require A set $S$ of $n$ positive integers and an upper bound integer $t$.
    	\Ensure The set $Z \subseteq (\mathcal{S}_t(S))^k \times [n]^k$ of all $k$-tuples of subset sums occurring from disjoint subsets of $S$ up to $t$, along with their respective cardinality information.

        \If{$S = \{s\}$}
            \Ret{$\{ 0 \}^{2k} \cup \{(s,\underbrace{0,\ldots,0}_{k-1}, 1, \underbrace{0,\ldots,0}_{k-1})\} \cup \ldots \cup \{(\underbrace{0,\ldots,0}_{k-1},s,\underbrace{0,\ldots,0}_{k-1},1)\}$}
        \EndIf
        \State $T \gets$ an arbitrary subset of $S$ of size $\lfloor \frac{n}{2} \rfloor$
        \Ret{$\texttt{DisjointSC}(T,t) \oplus \texttt{DisjointSC}(S \setminus T,t)$}
    \end{algorithmic}
\end{algorithm}

\begin{algorithm}[htb]
    \caption{$\texttt{DisjointSS}(Z,t)$}
    \label{alg:koil2}
    \begin{algorithmic}[1]
        \Require A set $Z$ of $n$ positive integers and an upper bound integer $t$.
    	\Ensure The set $S \subseteq (\mathcal{S}_t (Z))^k$ of all $k$-tuples of subset sums up to $t$ occurring from disjoint subsets of $Z$.

        \State $b \gets \lfloor \sqrt[k+1]{n^k \log n} \rfloor$
        \For{$l \in [b-1]$}
            \State $S_l \gets Z \cap \{ x \in \mathbb{N} \mid x \equiv l \Mod{b} \}$
            \State $Q_l \gets \{ \lfloor x/b \rfloor \mid x \in S_l \}$
            \State $\mathcal{R}(Q_l) \gets \texttt{DisjointSC}(Q_l, \lfloor t/b \rfloor)$
            \State $R_l \gets \{ (z_1 b + j_1 l, \ldots, z_k b + j_k l) \mid (z_1, \ldots, z_k, j_1, \ldots, j_k) \in \mathcal{R} (Q_l) \}$ \label{line:koil}
        \EndFor
        \Ret{ $R_0 \oplus_t \cdots \oplus_t R_{b-1}$}
    \end{algorithmic}
\end{algorithm}

%% file: sections/3_2_kss_bringmann.tex
\subsection{Solving \kSS\ in randomised $\tilde{O}(n + t^k)$ time}
\label{subsec:bringmann}

We will show that one can decide \kSS\ in $\tilde{O}(n + t^k)$ time, where $t = \max \{ t_1, \ldots, t_k \}$, by building upon the techniques used in~\cite{Bringmann17}.
In particular, we will present an algorithm that successfully detects, with high probability, whether there exist $k$ disjoint subsets each summing up to $t_i$ respectively.
In comparison to the algorithm of~\cite{Bringmann17}, a couple of modifications are required, which we will first briefly summarise prior to presenting the complete algorithm.

\begin{itemize}
    \item In our version of \ColorCoding, the number of repetitions of random partitioning is increased, without asymptotically affecting the complexity of the algorithm, since it remains $O(\log (1/\delta))$.

    \item In our version of \ColorCoding, after each partition of the elements, we execute the FFT operations on the $k$-modified characteristic polynomials of the resulting sets.
    Thus, for each element $s_i$ we introduce $k$ points $(s_i, 0, \ldots, 0)$, $\ldots$, $(0, \ldots, s_i)$, represented by polynomial $x^{s_i}_1 + \ldots + x^{s_i}_k$.
    Hence \ColorCoding\ returns a set of points, each of which corresponds to $k$ sums, realisable by disjoint subsets of the initial set.
    
    \item Algorithm \ColorCodingLayer\ needs no modification, except that now the FFT operations concern sets of points and not sets of integers, hence the complexity analysis differs.
    Additionally, the algorithm returns a set of points, each of which corresponds to realisable sums by disjoint subsets of the $l$-layer input instance.
    
    \item The initial partition of the original set to $l$-layer instances remains as is, and the FFT operations once more concern sets of points instead of sets of integers.
\end{itemize}

\subsubsection{Small cardinality solutions}

We will first describe the modified procedure \ColorCoding\ for solving \kSS\ if the solution size is small, i.e., an algorithm that, given an integer $c$, finds $k$-tuples of sums $(\Sigma (Y_1), \ldots, \Sigma (Y_k))$, where $\Sigma (Y_i) \leq t$ and $Y_i \subseteq Z$ are disjoint subsets of input set $Z$ of cardinality at most $c$, and determines with high probability whether there exists a tuple $(t_1, \ldots, t_k)$ among them, for some given values $t_i$.

We randomly partition our initial set $Z$ to $c^2$ subsets $Z_1, \ldots, Z_{c^2}$, i.e., we assign each $z\in Z$ to a set $Z_i$ where $i$ is chosen independently and uniformly at random from $\{1, \ldots, c^2\}$.
We say that this random partition \emph{splits} $Y\subseteq Z$ if $\rvert Y \cap Z_i\rvert  \leq 1, \forall i$.
If such a split occurs, the set returned by \ColorCoding\ will contain\footnote{In this context, ``contain'' is used to denote that $\Sigma(Y) = s_i$ for some $i$ in a $k$-tuple $s = (s_1, \ldots, s_k) \in S$, where $S$ is the resulting set of \ColorCoding.} $\Sigma(Y)$.
Indeed, by choosing the element of $Y$ for those $Z_i$ that $\rvert Y \cap Z_i\rvert =1$ and $0$ for the rest, we successfully generate $k$-tuples containing $\Sigma (Y)$ through the pairwise sum operations.
The algorithm returns only valid sums, since no element is used more than once in each sum, because each element is assigned uniquely to a $Z_i$ for each distinct partition.

Our intended goal is thus, for any $k$ random disjoint subsets, to have a partition that splits them all.
Such a partition allows us construct a $k$-tuple consisting of all their respective sums through the use of FFT.
Hence, the question is to determine how many random partitions are required to achieve this with high probability.

The answer is obtained by observing that the probability to split a subset $Y$ is the same as having $\lvert Y \rvert$ different balls in $\lvert Y \rvert$ distinct bins, when throwing $\lvert Y \rvert$ balls into $c^2$ different bins, as mentioned in~\cite{Bringmann17}.
Next, we compute the probability to split $k$ random disjoint subsets $Y_1, \ldots, Y_k$ in the same partition.

The probability that a split occurs at a random partition for $k$ random disjoint subsets $Y_1, \ldots, Y_k \subseteq Z$ is
\begin{align*}
    \Pr[\text{all } Y_i \text{ are split}] =
    \prod_{i=1}^k \Pr[Y_i\text{ is split}]&=\\
    \frac{c^2-1}{c^2} \cdots \frac{c^2-(\lvert Y_1 \rvert -1)}{c^2} \cdots
    \frac{c^2-1}{c^2} \cdots \frac{c^2-(\lvert Y_k \rvert -1)}{c^2} &\geq \\
    \left( \frac{c^2-(\lvert Y_1 \rvert -1)}{c^2}\right)^{\lvert Y_1 \rvert} \cdots
    \left( \frac{c^2-(\lvert Y_k \rvert -1)}{c^2}\right)^{\lvert Y_k \rvert} &\geq \\
    \left( 1 - \frac{1}{c} \right)^{c} \cdots \left( 1 - \frac{1}{c} \right)^{c} \geq
    \left( \frac{1}{2} \right)^2 \cdots \left( \frac{1}{2} \right)^2 &= \frac{1}{4^k}
\end{align*}

Hence, for $\beta = 4^k / (4^k - 1)$, $r = \lceil \log_{\beta} (1/\delta) \rceil$ repetitions yield the desired success probability of  $1 - (1 - 1 / 4^k)^r \geq 1 - \delta$.
In other words, after $r$ random partitions, for any $k$ random disjoint subsets $Y_1, \ldots, Y_k$, there exists, with probability at least $1 - \delta$, a partition that splits them all.

\begin{algorithm}[htb]
    \caption{$\ColorCoding(Z,t,c,\delta)$}
    \label{alg:ColorCoding}
    \begin{algorithmic}[1]
        \Require A set $Z$ of positive integers, an upper bound $t$, a size bound $c \geq 1$ and an error probability $\delta >0$.
    	\Ensure A set $S \subseteq (\mathcal{S}_t (Z))^k$ containing any $( \Sigma (Y_1), \ldots,  \Sigma (Y_k))$ with probability at least $1 - \delta$, where $Y_1, \ldots, Y_k \subseteq Z$ disjoint subsets with $\Sigma (Y_1), \ldots, \Sigma (Y_k) \leq t$ and $\lvert Y_1 \rvert, \ldots, \lvert Y_k \rvert \leq c$.

        \State $S \gets \emptyset$
        \State $\beta \gets 4^k / (4^k - 1)$
        \For{$j = 1, \ldots, \lceil \log_{\beta} (1 / \delta)\rceil$}
            \State randomly partition $Z = Z_1 \cup Z_2 \cup \cdots \cup Z_{c^2}$
            \For{$i = 1, \ldots, c^2$}
                \State $Z'_i \gets (Z_i \times \{ 0 \}^{k - 1}) \cup \ldots \cup (\{ 0 \}^{k - 1} \times Z_i)$
            \EndFor
            \State $S_j \gets Z'_1 \oplus_t \cdots \oplus_t Z'_{c^2}$
            \State $S \gets S \cup S_j$
        \EndFor
        \Ret{$S$}
    \end{algorithmic}
\end{algorithm}

\begin{lemma}
Given a set $Z$ of positive integers, a sum bound $t$, a size bound $c \geq 1$ and an error probability $\delta >0$, \ColorCoding($Z,t,k,\delta$) returns a set $S$ consisting of any $k$-tuple $(\Sigma (Y_1), \ldots, \Sigma (Y_k))$ with probability at least $1 - \delta$, where $Y_1, \ldots, Y_k \subseteq Z$ are disjoint subsets with $\Sigma (Y_1), \ldots, \Sigma (Y_k) \leq t$ and $\lvert Y_1 \rvert, \ldots, \lvert Y_k \rvert \leq c$, in $O(c^2 \log (1 / \delta) t^k \log t)$ time. 
\end{lemma}

\begin{proof}
As we have already explained, if there exist $k$ disjoint subsets $Y_1, \ldots, Y_k \subseteq Z$ with $\Sigma (Y_1), \ldots, \Sigma (Y_k) \leq t$ and $\lvert Y_1 \rvert, \ldots, \lvert Y_k \rvert \leq c$, our algorithm guarantees that with probability at least $1 - \delta$, there exists a partition that splits them all.
Subsequently, the FFT operations on the corresponding points produces the $k$-tuple.

\medskip

\noindent
\emph{Complexity.}
The algorithm performs $O(\log{(1/\delta)})$ repetitions.
To compute a pairwise sum of $k$ variables up to $t$, $O(t^k \log t)$ time is required.
In each repetition, $c^2$ pairwise sums are computed.
Hence, the total complexity of the algorithm is $O(c^2 \log (1 / \delta) t^k \log t)$.
\end{proof}

\subsubsection{Solving \kSS\ for $l$-layer instances}

In this part, we will prove that we can use the algorithm \ColorCodingLayer\ from~\cite{Bringmann17}, to successfully solve \kSS\ for $l$-layer instances, defined below.

For $l \geq 1$, we call $(Z,t)$ an \emph{$l$-layer instance} if $Z \subseteq [t / l, 2t / l]$ or $Z \subseteq [0, 2t / l]$ and $l \geq n$.
In both cases, $Z\subseteq [0,2t/l]$ and for any $Y \subseteq Z$ with $\Sigma (Y) \leq t$, we have $\lvert Y \rvert \leq l$.
The algorithm of~\cite{Bringmann17} successfully solves the \SubS\ problem for $l$-layer instances.
We will show that by calling the modified \ColorCoding\ algorithm and modifying the FFT operations so that they concern sets of points, the algorithm successfully solves \kSS\ in such instances.

\begin{algorithm}[htb]
    \caption{$\ColorCodingLayer(Z,t,l,\delta)$}
    \label{alg:ColorCodingLayer}
    \begin{algorithmic}[1]
        \Require An $l$-layer instance $(Z,t)$ and an error probability $\delta \in (0, 1 / 2^{k+1}]$.
    	\Ensure A set $S \subseteq (\mathcal{S}_t (Z))^k$ containing any $( \Sigma (Y_1), \ldots,  \Sigma (Y_k))$ with probability at least $1 - \delta$, where $Y_1, \ldots, Y_k \subseteq Z$ disjoint subsets with $\Sigma (Y_1), \ldots, \Sigma (Y_k) \leq t$.

        \lIf{$l < \log (l / \delta)$}{\Return $\ColorCoding(Z, t, l, \delta)$}
        \State $m \gets l / \log(l / \delta)$ rounded up to the next power of 2
        \State randomly partition $Z = Z_1 \cup Z_2 \cup \cdots \cup Z_m$
        \State $\gamma \gets 6 \log (l / \delta)$
        \For{$j = 1, \ldots, m$}
            \State $S_j \gets \texttt{ColorCoding}(Z_j, 2 \gamma t / l, \gamma, \delta / l)$ \label{line:bringmann}
        \EndFor
        \For{$h=1,\ldots,\log m$} \Comment{combine $S_j$ in a binary-tree-like way}
            \For{$j=1,\ldots,m/2^h$}
                \State $S_j \gets S_{2j-1} \oplus_{2^h \cdot 2\gamma t/l} S_{2j}$
            \EndFor
        \EndFor
        \Ret{$S_1 \cap [t]^k$}

    \end{algorithmic}
\end{algorithm}

We will now prove the correctness of the algorithm.
Let $X^1, \ldots, X^k \subseteq Z$ be disjoint subsets with $\Sigma(X^1), \ldots, \Sigma(X^k) \leq t$. By~\cite[Claim~3.3]{Bringmann17}, we have that $\Pr[\lvert Y_i \rvert \geq 6 \log{(l/\delta)}]\leq \delta / l$, where $Y_i = Y \cap Z_i$, for any $Y \subseteq Z$ with at most $l$ elements.
Hence, the probability that $\lvert X^j_i \rvert \leq 6\log{(l/\delta)}$, for all $i = 1, \ldots, m$ and $j = 1, \ldots, k$ is
\begin{align*}
    \Pr \left[\bigwedge_{i=1}^{m} \bigwedge_{j=1}^{k} \lvert X^j_i \rvert \leq 6\log{(l/\delta)} \right] &\geq
    1 - \left( \sum_{i = 1}^{m} \sum_{j=1}^{k} \Pr \left[ \lvert X^j_i \rvert > 6 \log(l/\delta) \right] \right)\\
    &\geq 1 - k m \delta / l.
\end{align*}
%
\ColorCoding\ computes $(\Sigma(X^1_i), \ldots, \Sigma(X^k_i))$ with probability at least $1-\delta$.
This happens for all $i$ with probability at least $1 - m \delta / l$.
Then, combining the resulting sets indeed yields the $k$-tuple $(\Sigma(X^1), \ldots, \Sigma(X^k))$.
The total error probability is at most $(k+1) m \delta / l$.
Assume that $\delta \in (0,1 / 2^{k+1}]$.
Since $l \geq 1$ and $\delta \leq 1 / 2^{k + 1}$, we have $\log{(l/\delta)} \geq (k + 1)$.
Hence, the total error probability is bounded by $\delta$.
This gives the following.

\begin{lemma}
Given an $l$-layer instance $(Z,t)$, upper bound $t$ and error probability $\delta \in (0,1/2^{k+1}]$, \ColorCodingLayer($Z,t,l,\delta$) solves \kSS\ with sum at most $t$ in time $O \left( t^k \log t \frac{\log^{k + 2} (l/\delta)}{l^{k - 1}} \right)$ with probability at least $1 - \delta$.
\end{lemma}

\noindent
\emph{Complexity.}
The time required to compute the sets of $k$-tuples $S_1, \ldots, S_m$ by calling \ColorCoding\ is
\begin{align*}
    O(m \gamma^2 \log (l/\delta) (\gamma t / l)^k \log (\gamma t / l)) &=
    O \left( \frac{\gamma^{k + 2}}{l^{k - 1}} t^k \log t \right)\\
    &=
    O \left( \frac{\log^{k + 2} (l/\delta)}{l^{k - 1}} t^k \log t \right).
\end{align*}
Combining the resulting sets costs
\begin{align*}
    O \left( \sum_{h=1}^{\log m} \frac{m}{2^h} (2^h \gamma t/l)^k \log (2^h \gamma t/l) \right) &=
    O \left(\sum_{h=1}^{\log m} \frac{2^{h (k - 1)}}{m^{k - 1}} t^k \log t \right)\\
    &= O \left( \frac{t^k \log t}{m^{k-1}} \sum_{h=1}^{\log m} \left( 2^{k - 1} \right)^h \right)\\
    &= O \left(\frac{t^k \log t}{m^{k-1}} \left( 2^{k - 1} \right)^{\log m} \right)\\
    &= O(t^k \log t),
\end{align*}
since for $c > 1$, we have that $O \left(\sum_{k=0}^n c^k \right) = O(c^n)$, which is dominated by the computation of $S_1, \ldots, S_m$.

\smallskip
\noindent
Hence, \ColorCodingLayer\ has total complexity $O \left( t^k \log t \cdot \frac{\log^{k + 2} (l/\delta)}{l^{k - 1}} \right)$.

\subsubsection{General Case}

It remains to show that for every instance $(Z,t)$, we can construct $l$-layer instances and take advantage of \ColorCodingLayer\ in order to solve \kSS\ for the general case.
This is possible by partitioning set $Z$ at $t/2^i$ for $i=1, \ldots, \lceil \log n \rceil -1$.
Thus, we have $O(\log n)$ $l$-layers $Z_1, \ldots, Z_{\lceil \log n\rceil}$.
On each layer we run \ColorCodingLayer\ and subsequently combine the results using pairwise sums.

\begin{algorithm}[htb]
    \caption{$\texttt{kSubsetSum}(Z,\delta, t)$}
    \label{alg:kss_bringmann}
    \begin{algorithmic}[1]
        \Require A set of positive integers $Z$ of cardinality $n$, an upper bound $t$ and an error probability $\delta$.
    	\Ensure A set $S \subseteq (\mathcal{S}_t (Z))^k$ containing any $( \Sigma (Y_1), \ldots,  \Sigma (Y_k))$ with probability at least $1 - \delta$, where $Y_1, \ldots, Y_k \subseteq Z$ disjoint subsets with $\Sigma (Y_1), \ldots, \Sigma (Y_k) \leq t$.
        
        \State partition $Z$ into $Z_i \gets Z \cap (t / 2^i, t / 2^{i-1}]$ for $i = 1, \ldots, \lceil \log n \rceil - 1$ and $Z_{\lceil \log n \rceil} \leftarrow Z \cap [0, t / 2^{\lceil \log n \rceil - 1}]$
        \State $S \gets \emptyset$
        \For{$i = 1, \ldots, \lceil \log n \rceil$}
            \State $S_i \gets \ColorCodingLayer(Z_i, t, 2^i, \delta / \lceil \log n \rceil)$
            \State $S \gets S \oplus_t S_i$
        \EndFor
        \Ret{$S$}
    \end{algorithmic}
\end{algorithm}

\smallskip
\noindent
We will now prove the main theorem of this section.

\begin{theorem} \label{theorem:bringmann}
Given a set $Z \subseteq \mathbb{N}$ of size $n$ and targets $t_1, \ldots, t_k$, one can decide \kSS\ in $\tilde{O}(n + t^k)$ time w.h.p., where $t = \max \{ t_1, \ldots, t_k \}$.
\end{theorem}

\begin{proof}
Let $X^1, \ldots, X^k \subseteq Z$ be $k$ disjoint subsets with $\Sigma (X^1), \ldots, \Sigma (X^k) \leq t$, and $X^j_i = X^j \cap Z_i$, for $j = 1, \ldots, k$ and $i = 1, \ldots, \lceil \log n \rceil$.
Each call to \ColorCodingLayer\ returns a $k$-tuple $(\Sigma (X^1_i), \ldots, \Sigma (X^k_i))$ with probability at least $1 - \delta / \lceil \log n \rceil$, hence the probability that all calls return the corresponding $k$-tuple is
\begin{align*}
    \Pr \left[\text{\ColorCodingLayer\ returns } (\Sigma (X^1_i), \ldots, \Sigma (X^k_i)), \forall i \right] &=
    1 - \Pr[\text{some call fails}]\\
    &\geq 1 - \sum_{i=1}^{\lceil \log n\rceil} \frac{\delta}{\lceil \log n\rceil}
    = 1-\delta
\end{align*}
If all calls return the corresponding $k$-tuple, the algorithm successfully constructs the $k$-tuple $(\Sigma (X^1), \ldots, \Sigma (X^k))$.
Thus, with probability at least $1-\delta$, the algorithm solves \kSS.

\medskip

\noindent
\emph{Complexity.}
Reading the input requires $\Theta(n)$ time.
In each of the $\Theta (\log n)$ repetitions of the algorithm, we make a call to \ColorCodingLayer, plus compute a pairwise sum.
The computation of the pairwise sum requires $O(t^k \log t)$ time since it concerns $k$-tuples.
For each call to \ColorCodingLayer, we require $O \left( t^k \log t \frac{\log^{k + 2} \left( \frac{2^i \log n}{\delta} \right) }{2^{i (k - 1)}} \right)$ time.
Hence, the overall complexity is
\[
O \left( n + t^k \log t \log n + \sum_{i=1}^{\log n} t^k \log t \frac{\log^{k + 2} \left( \frac{2^i \log n}{\delta} \right) }{2^{i (k - 1)}} \right) =\tilde{O}(n + t^k).
\]
\end{proof}

%% file: sections/4_0_fixed_cardinality.tex
\section{\kSS\ with cardinality constraints}
\label{sec:kss_cardinality}
The algorithms presented in the previous section can be extended to a variation of the \kSS\ problem, where along with the target values, the cardinality of the respective solution sets is given as part of the input.

Therefore, consider $c_i$, $i = 1, \ldots, k$ such that $\sum c_i \leq n$, to be the desired cardinality of set $Z_i$, where $\Sigma(Z_i) = t_i$, and let $c = \max \{ c_1, \ldots, c_k \}$.
We seek to determine whether there exist disjoint subsets $Z_i \subseteq Z, i = 1, \ldots, k$ of sum $\Sigma (Z_i) = t_i$ and cardinality $\lvert Z_i \rvert = c_i$.

As it will become evident, the previously proposed algorithms for \kSS\ can be slightly modified in order to cope with this additional restriction.
The main idea lies on the addition of extra variables in the resulting polynomial, each of which represents the cardinality of one of the solution sets.
This results in a blowup of the order of $O(c^k)$ in the final complexity of the algorithms, since a variable is added for each of the $k$ subsets of the solution.

%% file: sections/4_1_fixed_koiliaris.tex
\subsection{A deterministic approach}
\label{subsec:koiliaris_cardinality}
In the case of the deterministic algorithm proposed in Subsection~\ref{subsec:koiliaris}, it suffices to modify the procedure in line~\ref{line:koil} of Algorithm~\ref{alg:koil2}.

Here, instead of simply reconstructing the actual partial sums obtained from elements belonging to the same congruence class,
it is necessary to additionally keep the relevant cardinality information.
Subsequently, the final FFT operation takes into account this information in order to verify the existence of a solution, thus requiring additional time due to the extra variables.

\begin{theorem}
Given a set $Z = \{ z_1, \ldots, z_n \}\subseteq \mathbb{N}$ of size $n$, targets $t_1,\dots ,t_k$ and cardinalities $c_1, \ldots, c_k$, Algorithm~\ref{alg:koil2_mod} can decide \kSS\ with cardinality constraints in time $\tilde{O} (\sqrt[k+1]{n^k c^{k^2}} t^k)$, where $t = \max \{t_1,\dots,t_k \}$ and $c = \max \{ c_1, \ldots, c_k \}$.
\end{theorem}

\begin{proof}
The correctness of the algorithm can be shown by analogous arguments as those used in  Lemma~\ref{lem:disjointness}.
In particular, in this case the subset cardinalities computed in each congruence class are kept in order to verify the final cardinality of the solution subsets obtained from the addition of elements belonging in different congruence classes.

\medskip

\noindent
\emph{Complexity.}
The overall complexity of the algorithm, stemming from the computation of subset sums inside the congruence classes and the combination of those sums, is
\begin{align*}
    O \left( \frac{n^k \log n}{b^k} t^k \log t + b t^k c^k \log (t c) \right) &=
    O \left( \frac{n^k \log n}{b^k} t^k \log t + b t^k c^k \log t \right)\\
    &= O \left( t^k \log t \left( \frac{n^k \log n}{b^k} + b c^k \right) \right)\\
    &= \tilde{O} \left(\sqrt[k+1]{n^k c^{k^2}} t^k \right),
\end{align*}
which is obtained by setting $b = \sqrt[k+1]{\frac{n^k \log n}{c^k}}$.

\end{proof}

\begin{algorithm}[htb]
    \caption{$\texttt{CardDisjointSS}(Z,t)$}
    \label{alg:koil2_mod}
    \begin{algorithmic}[1]
        \Require A set $Z$ of $n$ positive integers, an upper bound integer $t$ and a size bound $c \geq 1$.
    	\Ensure The set $S \subseteq (\mathcal{S}_t (Z) \times [c])^k$ of all $k$-tuples of subset sums up to $t$ along with their respective cardinality on each dimension, occurring from disjoint subsets of $Z$ with cardinality up to $c$.

        \State $b \gets \big\lfloor \sqrt[k+1]{\frac{n^k \log n}{c^k}} \big\rfloor$
        \For{$l \in [b-1]$}
            \State $S_l \gets Z \cap \{ x \in \mathbb{N} \mid x \equiv l \Mod{b} \}$
            \State $Q_l \gets \{ \lfloor x/b \rfloor \mid x \in S_l \}$
            \State $\mathcal{R}(Q_l) \gets \texttt{DisjointSC}(Q_l, \lfloor t/b \rfloor)$
            \State $R_l \gets \{ (z_1 b + j_1 l, j_1), \ldots, (z_k b + j_k l, j_k) \mid (z_1, \ldots, z_k, j_1, \ldots, j_k) \in \mathcal{R} (Q_l) \}$
        \EndFor
        \Ret{ $R_0 \oplus_t \cdots \oplus_t R_{b-1}$}
    \end{algorithmic}
\end{algorithm}

%% file: sections/4_2_fixed_bringmann.tex
\subsection{A randomised approach}
\label{subsec:bringmann_cardinality}
With slight modifications, it is possible to successfully adapt the algorithm of Subsection~\ref{subsec:bringmann} in order to consider the additional cardinality constraints of this section.

More specifically, there are two important remarks that lead to those modifications:
\begin{enumerate}
    \item It is necessary to verify the cardinalities of the involved subsets in the resulting polynomial.
    This can be achieved by introducing additional variables, similar to~\ref{subsec:koiliaris_cardinality}.

    \item Whereas in the algorithm of subsection~\ref{subsec:bringmann} there was a necessity of partitioning the input to layers in order to associate the cardinality of a subset with its target value, in this case such a procedure is redundant, since the upper bound on the cardinality of the sets is part of the input.

\end{enumerate}

\medskip
Following these remarks, we firstly modify the \texttt{ColorCoding} algorithm of the previous section, in order to additionally return cardinality information for the involved subsets.
This does not affect the correctness of the algorithm whatsoever and only incurs a blowup in the final complexity.

\begin{algorithm}[htb]
    \caption{$\texttt{CardColorCoding}(Z,t,c,\delta)$}
    \label{alg:ColorCoding_mod}
    \begin{algorithmic}[1]
        \Require A set $Z$ of positive integers, an upper bound $t$, a size bound $c \geq 1$ and an error probability $\delta >0$.
    	\Ensure A set $S \subseteq (\mathcal{S}_t (Z))^k \times [c]^k$ such that it contains any tuple $( \Sigma (Y_1), \ldots,  \Sigma (Y_k), \lvert Y_1 \rvert, \ldots, \lvert Y_k \rvert)$ with probability at least $1 - \delta$, where $Y_1, \ldots, Y_k \subseteq Z$ disjoint subsets of $Z$ with $\Sigma (Y_1), \ldots, \Sigma (Y_k) \leq t$ and $\lvert Y_1 \rvert, \ldots, \lvert Y_k \rvert \leq c$.

        \State $S \gets \emptyset$
        \State $\beta \gets 4^k / (4^k - 1)$
        \For{$j = 1, \ldots, \lceil \log_{\beta} (1 / \delta)\rceil$}
            \State randomly partition $Z = Z_1 \cup Z_2 \cup \cdots \cup Z_{c^2}$
            \For{$i = 1, \ldots, c^2$}
                \State $Z'_i \gets (Z_i \times \{ 0 \}^{k - 1} \times \{1\} \times \{0\}^{k-1}) \cup \ldots \cup (\{ 0 \}^{k - 1} \times Z_i \times \{0\}^{k-1} \times \{1\})$
            \EndFor
            \State $S_j \gets Z'_1 \oplus_t \cdots \oplus_t Z'_{c^2}$
            \State $S \gets S \cup S_j$
        \EndFor
        \Ret{$S$}
    \end{algorithmic}
\end{algorithm}

\noindent
\emph{Complexity.}
The algorithm performs $O(\log{(1/\delta)})$ repetitions.
To compute a pairwise sum of $k$ variables up to $t$ along with the respective cardinality of each dimension, $O(t^k c^k \log (t c)) = O(t^k c^k \log t)$ time is required.
In each repetition, $c^2$ pairwise sums are computed.
Hence, the total complexity of the algorithm is $O(c^{k + 2} \log (1 / \delta) t^k \log t)$.

\paragraph*{Algorithm description}
In this paragraph, we briefly describe the algorithm that successfully solves \kSS\ with cardinality constraints.
The proposed algorithm is called \texttt{CardColorCodingLayer} and is identical to Algorithm~\ref{alg:ColorCodingLayer} of Subsection~\ref{subsec:bringmann} with the following modifications:
\begin{enumerate}[label=(\alph*)]
    \item The calls to the \texttt{ColorCoding} algorithm on line~\ref{line:bringmann} actually consider the \texttt{CardColorCoding} algorithm previously described, hence the cardinality information is included.
    
    \item The FFT operations now involve polynomials of $2k$ variables, since an additional variable is considered for each subset, depicting its cardinality.
\end{enumerate}
\medskip

Thus, in contrast to Section~\ref{subsec:bringmann} where we had to run a \texttt{ColorCodingLayer} algorithm for each distinct layer, in this case, it suffices to run $\mathtt{CardColorCodingLayer}(Z, t, c, \delta)$.

\begin{theorem}
    Given a set $Z = \{ z_1, \ldots, z_n \}\subseteq \mathbb{N}$ of size $n$, targets $t_1,\dots ,t_k$ and cardinalities $c_1, \ldots, c_k$, as well as an error probability $\delta \in (0,1/2^{k+1}]$, Algorithm {\normalfont\texttt{CardColorCodingLayer}} can decide \kSS\ with cardinality constraints with probability at least $1 - \delta$ in time $\tilde{O} (n + t^k c^k)$, where $t = \max \{t_1,\dots,t_k \}$ and $c = \max \{ c_1, \ldots, c_k \}$.
\end{theorem}

\begin{proof}
The correctness analysis is not affected by the aforementioned modifications, thus we only have to compute the complexity of the modified algorithm.
\medskip

\noindent
\emph{Complexity.}
In order to read the input set, $O(n)$ time is required.
The cost of the $m$ calls to \texttt{CardColorCoding} is
\begin{align*}
    O(m \gamma^{k + 2} \log (c / \delta) (\gamma t / c)^k \log (\gamma t / c)) &=
    O \left( \frac{\gamma^{2k + 2}}{c^{k - 1}} t^k \log t \right)\\
    &=
    O \left( \frac{\log^{2k + 2} (c / \delta)}{c^{k - 1}} t^k \log t \right).
\end{align*}
Combining the resulting sets costs
\begin{align*}
    O \left( \sum_{h=1}^{\log m} \frac{m}{2^h} (2^h \gamma t / c)^k c^k \log (2^h \gamma t c / c) \right) &=
    O \left(\sum_{h=1}^{\log m} \frac{2^{h (k - 1)}}{m^{k - 1}} t^k c^k \log t \right)\\
    &= O \left( \frac{t^k c^k \log t}{m^{k-1}} \sum_{h=1}^{\log m} \left( 2^{k - 1} \right)^h \right)\\
    &= O \left(\frac{t^k c^k \log t}{m^{k-1}} \left( 2^{k - 1} \right)^{\log m} \right)\\
    &= O(t^k c^k \log t),
\end{align*}
since for $c > 1$, it holds that $O \left(\sum_{k=0}^n c^k \right) = O(c^n)$.

\smallskip
\noindent
Thus, the total complexity is
\[
    O \left( n + t^k \log t \left(c^k + \frac{\log^{2k+2} (c / \delta)}{c^{k-1}} \right) \right) =
    \Tilde{O} (n + t^k c^k)
\]
\end{proof}

%% file: sections/5_further_applications.tex
\section{Faster Algorithms for Multiple Subset Problems}
\label{sec:further}

The techniques developed in Sections~\ref{sec:kss} and~\ref{sec:kss_cardinality} can be further applied to give faster pseudopolynomial algorithms for the decision version of the problems \SSR, \kSSR\ and \MSS.
In this section we will present how these algorithms can be used to efficiently solve these problems.

The algorithms we previously presented result in a polynomial $P(x_1, \ldots, x_k)$ consisting of terms each of which corresponds to a $k$-tuple of realisable sums by disjoint subsets of the initial input set $Z$.
In other words, if there exists a term $x_1^{s_1} \ldots x_k^{s_k}$ in the resulting polynomial, then there exist disjoint subsets $Z_1, \ldots, Z_k \subseteq Z$ such that $\Sigma (Z_1) = s_1, \ldots, \Sigma (Z_k) = s_k$.
Hereinafter, when we refer to a solution of a \kSS\ input, we actually refer to this resulting polynomial, unless explicitly stated otherwise.

It is important to note that, while the deterministic algorithm of Subsection~\ref{subsec:koiliaris} returns a polynomial consisting of \emph{all} terms corresponding to such $k$-tuples of realisable sums by disjoint subsets, the randomised algorithm of Subsection~\ref{subsec:bringmann} does not.
However, that does not affect the correctness of the following algorithms, since it suffices to guarantee that the $k$-tuple corresponding to the optimal solution of the respective (optimisation) problem is included with high probability.
That indeed happens, since the resulting polynomial consists of \emph{any} viable term with high probability, as discussed previously.

\paragraph*{\SSR}
The first variation we will discuss is \SSR, which asks to determine, given a set $Z \subseteq \mathbb{N}$ of size $n$ and an upper bound $t$, what is the smallest ratio of sums between any two disjoint subsets $S_1,S_2 \subseteq Z$, where $\Sigma(S_1),\Sigma(S_2) \le t$.
This can be solved in deterministic $\tilde{O}(n^{2/3}t^2)$ time using the algorithm proposed in subsection \ref{subsec:koiliaris} by simply iterating over the terms of the final polynomial that involve both parameters $x_1$ and $x_2$ and checking the ratio of their exponents.
\SSR\ can also be solved with high probability in randomised $\tilde{O}(n + t^2)$ time using the algorithm proposed in Subsection \ref{subsec:bringmann} instead.

\paragraph*{\kSSR}
An additional extension is the \kSSR\ problem, which asks, given a set $Z \subseteq \mathbb{N}$ of size $n$ and $k$ bounds $t_1,\dots,t_k$, to determine what is the smallest ratio between the largest and smallest sum of any set of $k$ disjoint subsets $Z_1, \ldots, Z_k \subseteq Z$ such that $\Sigma (Z_i) \le t_i$.
Similar to \kSS, an interesting special case is when all $t_i$'s are equal, in which case we search for $k$ subsets that are as similar as possible in terms of sum.

Similarly, we can solve this in deterministic $\tilde{O}(n^{k/(k+1)} \cdot t^k)$ or randomised $\tilde{O} (n + t^k)$ time by using the corresponding algorithm to solve \kSS\ and subsequently iterating over the terms of the resulting polynomial that respect the corresponding bounds, and finally evaluating the ratio of the largest to smallest exponent.

\paragraph*{$k$\textsc{-way Number Partitioning}}
Another closely related problem called $k$\textsc{-way Number Partitioning} asks, given a set $Z \subseteq \mathbb{N}$ of size $n$, to partition its elements into $k$ subsets $Z_1, \ldots, Z_k$, whose sums $\Sigma (Z_i)$ are as similar as possible.
There are multiple different objective functions which may be used in order to define this concept of similarity.
For instance, one could ask for the minimum difference between the largest and smallest sums, the minimum ratio between the largest and smallest sums, the maximum possible smallest sum or alternatively the minimum possible largest sum.
Note that each of these objective functions may lead to a different solution~\cite{Korf10}.

In order to solve this problem for any of the previously mentioned objective functions, it suffices to solve \kSS\ for $t = \max(Z) + (\Sigma (Z) / k)$, where $\max(Z)$ denotes the largest element of $Z$, while only considering the terms $x_1^{s_1} \ldots x_k^{s_k}$ of the final polynomial for which $\sum s_i = \Sigma (Z)$.
This holds due to the fact that any candidate solution involves sets of subset sum at most equal to $t$, since if there exists a solution involving a subset $Z' \subseteq Z$ with $\Sigma (Z') > t$, that means that in the same solution, there exists a subset $Z'' \subseteq Z$ such that $\Sigma (Z'') < \Sigma(Z) / k$, on which case, removing any element from set $Z'$ and adding it to the set $Z''$ results in a potentially better solution, irrespective of the objective function used to define the concept of similarity.

\paragraph*{\MSS}
Finally, we consider the \MSS\ problem that asks, given a set $Z \subseteq \mathbb{N}$ of size $n$ and $k$ bounds $t_1,\dots,t_k$, to determine what is the maximum sum of sums of any set of $k$ disjoint subsets $Z_1, \ldots, Z_k$ of $Z$, such that $\Sigma (Z_i) \le t_i$.
This problem is a special case of the \textsc{Multiple Knapsack} problem and can also be seen as a generalisation of \kSS.
It should be clear that the same techniques as those e.g.\ used for \kSSR\ apply directly, leading to the same time complexity bounds of  $\tilde{O}(n^{k/(k+1)} \cdot t^k)$ deterministically and $\tilde{O} (n + t^k)$ probabilistically.

\paragraph*{Adding cardinality constraints}
Note that for all of the presented algorithmic problems, if specific cardinalities for the subsets involved in a solution are required, then we can obtain a solution by applying the algorithms of Section~\ref{sec:kss_cardinality} in an analogous manner.
In this case, the resulting polynomial is of $2k$ variables, $k$ of which are used to specify the cardinality of the subsets.
Thus, by checking those variables, it is possible to obtain a solution in this more restricted version of the problems.

%% file: sections/6_future_work.tex
\section{Future Work}

\paragraph*{Possible optimisations regarding FFT}
The algorithm of Subsection~\ref{subsec:koiliaris} as well as those presented in Section~\ref{sec:kss_cardinality} involve the computation of the possible sums by disjoint subsets along with their respective cardinality.
To do so, we extend the FFT operations to multiple variables, each representing either a possible subset sum or its cardinality.
Hence, for $k$ subsets and $n$ elements, we proceed with FFT operations on variables $x_1, \ldots, x_k, x_{k+1}, \ldots, x_{2k}$, where the exponents of $x_i, i \leq k$ are in $[t]$ for some given upper bound $t$, whereas the exponents of $x_i, i > k$ in $[n]$.
Notice however that each element is used only on a single subset, hence for a term $x^{s_1}_1 \ldots x^{s_k}_k x^{n_1}_{k+1} \ldots x^{n_k}_{2k}$ of the produced polynomial, it holds that $s_i \leq t$ and $\sum n_i \leq n$.
This differs substantially from our analysis, where we essentially only assume that $n_i \leq n$, which is significantly less strict.
Hence, a stricter complexity analysis may be possible on those FFT operations, resulting in a more efficient overall complexity for this algorithm.

Also notice that, in our proposed algorithms, each combination of valid sums appears $k!$ times.
This means that for every $k$ disjoint subsets $S_1, \ldots, S_k$ of the input set, there are $k!$ different terms in the resulting polynomial of the algorithm representing the combination of sums $\Sigma (S_1), \ldots, \Sigma (S_k)$.
This increase on the number of terms does not influence the asymptotic analysis of our algorithms, nevertheless can be restricted (e.g. by sorting and pruning) for better performance.
Additionally, one can limit the FFT operations to different bounds for each variable, resulting in slightly improved complexity analysis without changing the algorithms whatsoever.
In this paper, we preferred to analyse the complexity of the algorithms using $t = \max \{ t_1, \ldots, t_k \}$ for the sake of simplicity, but one can alternatively obtain time complexities of $\tilde{O} (n^{k / (k+1)} T)$ and $\tilde{O} (n + T)$ for the deterministic and the randomised algorithm respectively, where $T = \prod t_i$.
These can be further extended to the case where we additionally have cardinality constraints.

\paragraph*{Recovery of solution sets}
The algorithms introduced in this paper solve the \emph{decision version} of the \kSS\ problem.
In other words, their output is a binary response, indicating whether there exist $k$ disjoint subsets whose sums are equal to given values $t_1, \ldots, t_k$ respectively (and their corresponding cardinalities are equal to $c_1, \ldots, c_k$ in the cardinality constrained version of the problem).
An extension of these algorithms could involve the reconstruction of the $k$ solution subsets.
Koiliaris and Xu~\cite{KoiliarisX19} argue that one can reconstruct the solution set of \SubS\ with only polylogarithmic overhead.
That is possible by carefully extracting the \emph{witnesses} of each sum every time an FFT operation is happening.
These witnesses are actually the partial sums used to compute the new sum.
Thus, by reducing this problem to the \emph{reconstruction problem} as mentioned in~\cite{ALLT2011}, they conclude that it is possible to compute all the witnesses of an FFT operation without considerably increasing the complexity.
That is the case for one-dimensional FFT operations involving a single variable, so it may be possible to use analogous arguments for multiple variables.

\paragraph*{Extension of \SubS\ algorithm introduced by Jin and Wu}
Jin and Wu introduced an efficient $\tilde{O} (n + t)$ randomised algorithm for solving \SubS\ in~\cite{JinWu}.
This algorithm is much simpler than Bringmann's and actually has slightly better complexity.
It is interesting to research whether this algorithm can be extended to cope with \kSS\ (and the variations mentioned in Section~\ref{sec:further}), as was the case for Bringmann's, since that would result in a simpler alternative approach.